%% file: thinking_twice.tex
\documentclass[11pt]{article}

\pdfoutput=1 % necessary to submit to arXiv

\usepackage{graphicx}
\usepackage{amssymb, amsmath, amsthm}
\usepackage{subfigure}
\usepackage{setspace}
\usepackage{multirow}
\usepackage{fullpage}

\newtheorem{theorem}{Theorem}

\newtheorem{definition}[theorem]{Definition}
\newtheorem{lemma}[theorem]{Lemma}

\newcommand{\paren}[1]{\left({#1}\right)}
\newcommand{\braces}[1]{\left\{{#1}\right\}}
\newcommand{\bracks}[1]{\left[{#1}\right]}

\newcommand{\setnot}[2]{\braces{#1 \mbox{ }: \mbox{ } #2}}

\renewcommand{\phi}{\varphi}

\newcommand{\eps}{\varepsilon}

\DeclareMathOperator \pr {\mathbb{P}r}
\DeclareMathOperator \expect {\mathbb{E}}
\DeclareMathOperator \ranking {Ranking}
\DeclareMathOperator \rankingsim {RankingSimulate}

\title{
  On Revenue Maximization in \\ 
  Second-Price Ad Auctions
}

\author{
  Yossi~Azar\thanks{
    \texttt{azar@tau.ac.il},
    Microsoft Research, Redmond and Tel-Aviv University.
  } \and
  Benjamin~Birnbaum\thanks{
    \texttt{birnbaum@cs.washington.edu}, University of Washington.
    Supported by an NSF Graduate Research Fellowship.
  } \and
  Anna~R.~Karlin\thanks{
    \texttt{karlin@cs.washington.edu}, University of Washington.
    Supported by NSF Grant CCF-0635147 and a grant from Yahoo!\ Research.
  } \and
  C.~Thach Nguyen\thanks{
    \texttt{ncthach@cs.washington.edu}, University of Washington.
    Supported by NSF Grant CCF-0635147 and a grant from Yahoo!\ Research.
  }
}

\begin{document}

\maketitle

\input{abstract}
\thispagestyle{empty}
\newpage
\setcounter{page}{1}
\input{introduction}
\input{model}

\input{flexible}

\input{offline2pm}
\input{online2pm}

\input{conclusion}

\bibliographystyle{alpha}
\bibliography{thinking_twice}

\appendix
\input{appendix_models}

\input{appendix_kvv}

\end{document}

%% file: abstract.tex
\begin{abstract}
  Most recent papers addressing the algorithmic problem of allocating
  advertisement space for keywords in sponsored search auctions assume
  that pricing is done via a first-price auction, which does not
  realistically model the Generalized Second Price (GSP) auction used
  in practice.  Towards the goal of more realistically modeling these
  auctions, we introduce the {\em Second-Price Ad Auctions} problem,
  in which bidders' payments are determined by the GSP mechanism.  We
  show that the complexity of the Second-Price Ad Auctions problem is
  quite different than that of the more studied First-Price Ad
  Auctions problem.  First, unlike the first-price variant, for which
  small constant-factor approximations are known, it is NP-hard to
  approximate the Second-Price Ad Auctions problem to any non-trivial
  factor.  Second, this discrepancy extends even to the $0$-$1$
  special case that we call the {\em Second-Price Matching} problem
  (2PM).  In particular, offline 2PM is APX-hard, and for online 2PM
  there is no deterministic algorithm achieving a non-trivial
  competitive ratio and no randomized algorithm achieving a
  competitive ratio better than $2$.  This stands in contrast to the
  results for the analogous special case in the first-price model, the
  standard bipartite matching problem, which is solvable in polynomial
  time and which has deterministic and randomized online algorithms
  achieving better competitive ratios.  On the positive side, we
  provide a 2-approximation for offline 2PM and a 5.083-competitive
  randomized algorithm for online 2PM.  The latter result makes use of
  a new generalization of a classic result on the performance of the
  ``Ranking'' algorithm for online bipartite matching.
\end{abstract}

%% file: introduction.tex
\newcommand \ourproblem {Second-Price Ad Auctions}
\newcommand \SecondPM {Second-Price Matching}
\section{Introduction}

The rising economic importance of online sponsored search advertising
has led to a great deal of research focused on developing its
theoretical underpinnings.  (See, e.g.,~\cite{Lahaie07} for a survey).
Since search engines such as Google, Yahoo!~and Bing depend on
sponsored search for a significant fraction of their revenue, a key
problem is how to optimally allocate ads to keywords (user searches)
so as to maximize search engine
revenue~\cite{Abrams07,Andelman04,Azar08,Buchbinder07,Chakrabarty08,Devanur09,Goel08a,Goel08b,Mahdian07,Mehta07,Srinivasan08}.
Most of the research on the dynamic version of this problem assumes
that once the participants in each keyword auction are determined, the
pricing is done via a first-price auction; in other words, bidders pay
what they bid. This does not realistically model the standard
mechanism used by search engines, called the Generalized Second Price
mechanism (GSP) \cite{Edelman07,Varian07}.

In an attempt to model reality more closely, we study the {\em
  \ourproblem} problem, which is the analogue of the above allocation
problem when bidders' payments are determined by the GSP mechanism.
As in other
work~\cite{Azar08,Buchbinder07,Chakrabarty08,Mehta07,Srinivasan08}, we
make the simplifying assumption that there is only one slot for each
keyword.
In this case, the GSP mechanism for a given keyword auction
reduces to a second-price auction -- given the participants in the
auction, it allocates the advertisement slot to the highest bidder,
charging that bidder the bid of the second-highest bidder.\footnote{
  This simplication, among others (see~\cite{Lahaie07}),
  leaves room to improve the accuracy of our model.  However, the
  hardness results clearly hold for the multi-slot case as well.
}

In the \ourproblem\ problem, there is a set of keywords $U$ and a set
of bidders $V$, where each bidder $v \in V$ has a known daily budget
$B_v$ and a non-negative bid $b_{u,v}$ for every keyword $u \in U$.
The keywords are ordered by their arrival time, and as each keyword
$u$ arrives, the algorithm (i.e., the search engine) must choose a
bidder to allocate it to.  The search engine is not required
to choose the highest-bidding bidder; in order to optimize the
allocation of bidders to keywords, search engines typically use a
``throttling'' algorithm that chooses which bidders to select to
participate in an auction for a given
keyword~\cite{Goel08a}.\footnote{ In this paper, we assume the search
  engine is optimizing over revenue although it is certainly
  conceivable that a search engine would consider other objectives.  }

In the previously-studied first-price version of the problem,
allocating a keyword to a bidder meant choosing a single bidder $v$
and allocating $u$ to $v$ at a price of $b_{u,v}$.  In the
\ourproblem\ problem, two bidders are selected instead of one.  Of
these two bidders, the bidder with the higher bid (where bids are
always reduced to the minimum of the actual bid and bidders' remaining
budgets) is allocated that keyword's advertisement slot at the price
of the other bid.  (In the GSP mechanism for $k$ slots, $k + 1$
bidders are selected, and each of the top $k$ bidders pays the bid of
the next-highest bidder.)

This process results in an allocation and pricing of the advertisement
slots associated with each of the keywords. The goal is to select the
bidders participating in each auction to maximize the total profit
extracted by the algorithm. For an example instance of this problem,
see Figure~\ref{fig:2paa_example}.

\begin{figure}[!b]
  \begin{center}
    %\begin{tabular}{cccc}
    %  \subfigure[]{\scalebox{0.6}{\label{fig:2paa_example:a}
    %      \includegraphics{2paa_example_a}
    %    } 
    %  } &
    %  \subfigure[]{\scalebox{0.6}{\label{fig:2paa_example:b}
    %      \includegraphics{2paa_example_b}
    %    }
    %  } &
    %  \subfigure[]{\scalebox{0.6}{\label{fig:2paa_example:c}
    %      \includegraphics{2paa_example_c}
    %    } 
    %  } &
    %  \subfigure[]{\scalebox{0.6}{\label{fig:2paa_example:d}
    %      \includegraphics{2paa_example_d}
    %    } 
    %  } 
    %\end{tabular}
    \scalebox{1.0}{\includegraphics{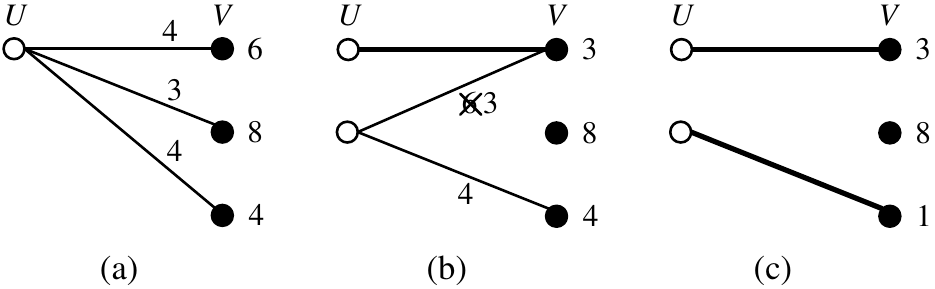}}
  \end{center}
  \caption{ \small
    An example of the \ourproblem\ problem: the nodes in $U$
    are keywords and the nodes in $V$ are bidders.  The number
    immediately to the right of each bidder represents its remaining
    budget, and the number next to each edge connecting a bidder to a
    keyword represents the bid of that bidder for that keyword.  (a)
    shows the situation when the first keyword arrives.  For this
    keyword, the search engine selects the first bidder, whose bid is
    4, and the second bidder, whose bid is 3.  The keyword is
    allocated to the first bidder at a price of 3, thereby reducing
    that bidder's budget by 3.  (b) shows the situation when the
    second keyword arrives.  The bid of the first bidder for that
    keyword is adjusted to the minimum of its original bid, $6$, and
    its remaining budget, $3$.  Then the first and the third bidders
    are selected, and the keyword is allocated to the third bidder at
    a price of $3$.} \label{fig:2paa_example} 
\end{figure}

\subsection{Our Results}\label{sec:results}

We begin by considering the {\em offline} version of the
\ourproblem\ problem, in which the algorithm knows all of the original
bids of the bidders (Section~\ref{sec:flexible}).  Our main result
here is that it is NP-hard to approximate the optimal solution to this
problem to within a factor better than $\Omega(m)$, where $m$ is the
number of keywords, even when the bids are small compared to budgets.
This strong inapproximability result is matched by the trivial
algorithm that selects the single keyword with the highest second-best
bidder and allocates only that keyword to its top two bidders.  It
stands in sharp contrast to the standard First-Price Ad Auctions
problem, for which there is a 4/3-approximation to the offline
problem~\cite{Chakrabarty08,Srinivasan08} and an $e/(e-1)$-competitive
algorithm to the online problem when bids are small compared to
budgets~\cite{Buchbinder07,Mehta07}.

We then turn our attention to a theoretically appealing special case
that we call {\em \SecondPM}.  In this version of the problem, all
bids are either 0 or 1 and all budgets are 1.  This can be thought of
as a variant on maximum bipartite matching in which the input is a
bipartite graph $G = (U \cup V, E)$, and the vertices in $U$ must be
matched, in order, to the vertices in $V$ such that the profit of
matching $u \in U$ to $v \in V$ is $1$ if and only if there is at
least one additional vertex $v' \in V$ that is a neighbor of $u$ and
is unmatched at that time.  One can justify the second-price version
of the problem by observing that when we sell an item, we can only
charge the full value of the item when there is more than one
interested buyer.\footnote{ A slightly more amusing motivation is to
  imagine that the two sets of nodes represent boys and girls and the
  edges represent mutual interest, but a girl is only interested in a
  boy if another girl is also actively interested in that boy.  }

Recall that the first-price analogue to the \SecondPM\ problem, the
maximum bipartite matching problem, can be solved optimally in
polynomial time.  The online version has a trivial 2-competitive
deterministic greedy algorithm and an $e/(e-1)$-competitive randomized
algorithm due to Karp, Vazirani and Vazirani~\cite{Karp90}, both of
which are best possible.

In contrast, we show that the \SecondPM\ problem is APX-hard
(Section~\ref{sec:hardness_offline2pm}).  We also give a
2-approximation algorithm for the offline problem
(Section~\ref{sec:approx_offline2pm}). We then turn to the online
version of the problem. Here, we show that no deterministic online
algorithm can get a competitive ratio better than $m$, where $m$ is
the number of keywords in the instance, and that no randomized online
algorithm can get a competitive ratio better than 2
(Section~\ref{sec:lower_online2pm}).  On the other hand, we present a
randomized online algorithm that achieves a competitive ratio of $2
\sqrt{e}/(\sqrt{e}-1) \approx 5.08$
(Section~\ref{sec:rand_upper_online2pm}).  To obtain this competitive
ratio, we prove a generalization of the result due to Karp, Vazirani,
and Vazirani~\cite{Karp90} and Goel and Mehta~\cite{Goel08b} that the
{\em Ranking} algorithm for online bipartite matching achieves a
competitive ratio of $e/(e-1)$.

\begin{table}
\small
	\begin{center}
	\begin{tabular}{|p{1.5cm}|p{2.1cm}|p{2cm}|p{4.3cm}|p{3cm}|c|}
		\hline
			\multirow{2}{*}{}
			&	\multicolumn{2}{c|}{Offline} & \multicolumn{2}{c|}{Online} \\
		%\hline
		  & Upper bound & Lower bound & Upper bound & Lower bound \\
		\hline 
			1PAA 			& $4/3$~\cite{Chakrabarty08,Srinivasan08} & 
								$16/15$~\cite{Chakrabarty08} & 
								$e/(e-1)^*$~\cite{Buchbinder07,Mehta07} or $2$~\cite{Lehmann06} 
								& $e/(e-1)$ \cite{Mehta07,Karp90}\\
		\hline
			2PAA			& $O(m)$ & $\Omega(m)$ & - & - \\
		\hline
			Matching	&\multicolumn{2}{c|}{poly-time alg.} &
								$e/(e-1)$~\cite{Karp90,Goel08b} & $e/(e-1)$~\cite{Karp90}\\
		\hline
			2PM				& $2$ & $364/363$ &
									$2\sqrt{e}/(\sqrt{e} - 1)\approx 5.083$ & $2$\\
		\hline
	\end{tabular}
	\vspace{.1in}
	\caption{ \small
          A summary of the results in this paper, compared to known
          results for the first-price case. The upper bound of 
	$e/(e-1)$ for Online 1PAA only holds when
	when the bids are small compared to the budgets.}
	\end{center}
\end{table}
\normalsize 

\subsection{Related Work}

As discussed above, the related First-Price Ad Auctions
problem\footnote{This problem has also been called the {\em Adwords}
  problem~\cite{Devanur09,Mehta07} and the {\em Maximum Budgeted
    Allocation} problem~\cite{Azar08,Chakrabarty08,Srinivasan08}.  It
  is an important special case of SMW
  \cite{Dobzinski06,Feige06b,Khot05,Lehmann06,Mirrokni08,Vondrak08},
  the problem of maximizing utility in a combinatorial auction in
  which the utility functions are submodular, and is also related to
  the Generalized Assignment Problem (GAP)
  \cite{Chekuri00,Feige06b,Fleischer06,Shmoys93}.  }  has received a
fair amount of attention.  Mehta et al.~\cite{Mehta07} present an
algorithm for the online version that achieves an optimal competitive
ratio of $e/(e-1)$ for the case when the bids are much smaller than
the budgets, a result also proved by Buchbinder et
al.~\cite{Buchbinder07}.  Under similar assumptions, Devanur and Hayes
show that when the keywords arrive according to a random permutation,
a $(1-\eps)$-approximation is possible~\cite{Devanur09}.  When there
is no restriction on the values of the bids relative to the budgets,
the best known competitive ratio is 2~\cite{Lehmann06}.  For the
offline version of the problem, a sequence of
papers~\cite{Lehmann06,Andelman04,Feige06b,Azar08,Srinivasan08,Chakrabarty08}
culminating in a paper by Chakrabarty and Goel, and independently, a
paper by Srinivasan, show that the offline problem can be approximated
to within a factor of $4/3$ and that there is no polynomial time
approximation algorithm that achieves a ratio better than 16/15 unless
$P=NP$~\cite{Chakrabarty08}.

The most closely related work to ours is the paper of Goel, Mahdian,
Nazerzadeh and Saberi \cite{Goel08a}, which builds on the work of
Abrams, Medelvitch, and Tomlin~\cite{Abrams07}.  Goel et al.\ look at
the online allocation problem when the search engine is committed to
charging under the GSP scheme, with multiple slots per keyword.  They
study two models, the ``strict'' and ``non-strict'' models, both of
which differ from our model even for the one slot case by allowing
bidders to keep bidding their orginal bid, even when their budget
falls below this amount.  Thus, in these models, although bidders are
not charged more than their remaining budget when allocated a keyword,
a bidder with a negligible amount of remaining budget can keep his
bids high indefinitely, and as long as this bidder is never allocated
another slot, this high bid can determine the prices other bidders pay
on many keywords.  Under the assumption that bids are small compared
to budgets, Goel et al.\ build on the linear programming formulation
of Abrams et al. to present an $e/(e-1)$-competitive algorithm for the
non-strict model and a 3-competitive algorithm for the strict model.

The significant, qualitative difference between these positive results
and the strong hardness we prove for our model suggests that these
aspects of the problem formulation are important.  We feel that our
model, in which bidders are not allowed to bid more than their
remaining budget, is more natural because it seems inherently unfair
that a bidder with negligible or no budget should be able to
indefinitely set high prices for other bidders.

%% file: model.tex
\section{Model and Notation}\label{sec:model}

We define the Second-Price Ad Auctions (2PAA) problem formally as
follows.  The input is a set of ordered keywords $U$ and bidders $V$.
Each bidder $v \in V$ has a budget $B_v$ and a nonnegative bid
$b_{u,v}$ for every keyword $u \in U$.  We assume that all of bidder
$v$'s bids $b_{u,v}$ are less than or equal to $B_v$.

Let $B_v(t)$ be the remaining budget of bidder $v$ immediately after
the $t$-th keyword is processed (so $B_v(0)= B_v$ for all $v$), and
let $b_{u,v}(t) = \min (b_{u,v}, B_v(t))$. (Both quantities are
defined inductively.)  A solution (or {\em second-price matching}) to
2PAA chooses for the $t$-th keyword $u$ a pair of bidders $v_1$ and
$v_2$ such that $b_{u,v_1}(t-1) \ge b_{u,v_2}(t-1)$, allocates the
slot for keyword $u$ to bidder $v_1$ and charges bidder $v_1$ a price
of $p(t)= b_{u,v_2}(t-1)$, the bid of $v_2$.  (We say that $v_1$ acts
as the \emph{first-price bidder} for $u$ and $v_2$ acts as the
\emph{second-price bidder} for $u$.)  The budget of $v_1$ is then
reduced by $p(t)$, so $B_{v_1}(t) = B_{v_1}(t-1) - p(t)$. For all
other bidders $v \ne v_1$, $B_{v}(t) = B_v(t-1)$.  The final value of
the solution is $\sum_t p(t)$, and the goal is to find a solution of
maximum value.

In the offline version of the problem, all of the bids are known to
the algorithm beforehand, whereas in the online version of the
problem, keyword $u$ and the bids $b_{u,v}$ for each $v \in V$ are
revealed only when keyword $u$ arrives, at which point the algorithm
must irrevocably map $u$ to a pair of bidders without knowing the bids
for the keywords that will arrive later.

The special case referred to as Second-Price Matching (2PM) is where
$b_{u,v}$ is either 0 or 1 for all $(u,v)$ pairs and $B_v=1$ for all
$v$.  We will think of this as the variant on maximum bipartite
matching (with input $G = (U \cup V, E)$) described in
Section~\ref{sec:results}.  Note that in 2PM, a keyword can only be
allocated for profit if its degree is at least two.  Therefore, we
assume without loss of generality that for all inputs of 2PM, the
degree of every keyword is at least two.

For an input to 2PAA, let $R_{min} = \min_{u,v} B_v / b_{u,v}$,
and let $m = |U|$ be the number of keywords.

%% file: flexible.tex
\section{Hardness of Approximation of 2PAA}\label{sec:flexible}

In this section, we present our main hardness result for the
\ourproblem\ problem.  For a constant $c \geq 1$, let
2PAA($c$) be the version of 2PAA in which we are promised that
$R_{min} \geq c$.
\begin{theorem}\label{thm:hardness_2paa}
Let $c \geq 1$ be a constant integer.  For any constant $c' > c$, it
is NP-hard to approximate 2PAA($c$) to a factor of $m/c'$.
\end{theorem}
\noindent
Hence, even when the bids are guaranteed to be smaller than the budget
by a large constant factor, it is NP-hard to approximate 2PAA to a
factor better than $\Omega(m)$.  After proving this result, we show in
Theorem~\ref{thm:matching_2paa} that this hardness is matched by a
trivial algorithm.
\begin{proof}
Fix a constant $c' > c$, and let $n_0$ be the smallest integer such
that for all $n \geq n_0$,
\begin{equation}\label{eqn:asymptotic1}
c' \cdot \frac{c(n^5 + n + 2)}{cn^2 + n + 2} \geq c(n^3 + cn^2 + n + 2)
\end{equation}
and 
\begin{equation}\label{eqn:asymptotic2}
\frac{n/2 + 1}{2} \geq c \enspace .
\end{equation}
Note that since $n_0$ depends only on $c'$, it is a constant.  

We reduce from PARTITION, in which the input is a set of $n \geq n_0$
items, and the weight of the $i$-th item is given by $w_i$.  If $W =
\sum_{i = 1}^n w_i$, then the question is whether there is a partition
of the items into two subsets of size $n/2$ such that the sum of the
$w_i$'s in each subset is $W/2$.  It is known that this problem (even
when the subsets must both have size $n/2$) is NP-hard~\cite{Garey79}.

Given an instance of PARTITION, we create an instance of 2PAA($c$) as
follows.  (This reduction is illustrated in
Figure~\ref{fig:hardness_2PMBA_small}.)
%Suppose that there is an $m/c'$-approximation algorithm to 2PAA($c$);
%we will show that constructing the following instance of 2PAA($c$)
%(illustrated in Figure~\ref{fig:hardness_2PMBA_small}) allows us to
%use the $m/c'$-approximation to solve the PARTITION instance:
\begin{itemize}
\item
First, create $n + 2$ keywords  $c_1, \ldots, c_n, e_1, e_2$.  Second,
create an additional set 
\begin{equation*}
G=\setnot{g_{i,k}}{1 \leq i \leq n^2 \mbox{ and } 1 \leq k \leq c}
\end{equation*}
of $cn^2$ keywords.  The keywords arrive in the order
\begin{equation*}
c_1, \ldots, c_n, e_1, e_2, g_{1,1}, \ldots, g_{1,c}, \ldots \ldots, g_{n^2,1}, \ldots, g_{n^2,c} \enspace .
\end{equation*}

\item
Create $n^2 + 4$ bidders $a, d_1, d_2, f, h_1, \ldots, h_{n^2}$.  Set
the budgets of $a$, $d_1$, and $d_2$ to $cW(1 + n/2)$.  Set the budget
of $f$ to $cW(n^3 + 1)$.  For $1 \leq i \leq n^2$, set the budget of
$h_i$ to $cWn^3$.

\item
For $1 \leq i \leq n$, bidders $a$, $d_1$, and $d_2$ bid $c(w_i + W)$
on keyword $c_i$.

\item
For $j \in \braces{1,2}$, bidder $d_j$ bids $cW$ on keyword $e_j$.
Bidder $f$ bids $cW/2$ on both $e_1$ and $e_2$.

\item
For $1 \leq i \leq n^2$ and $1 \leq k \leq c$, keyword $g_{i,k}$
receives a bid of $W(n^3 + 1)$ from bidder $f$ and a bid of $Wn^3$
from bidder $h_i$.
\end{itemize}
This reduction can clearly be performed in polynomial time.
Furthermore, it can easily be checked that (\ref{eqn:asymptotic2})
implies that no bidder bids more than $1/c$ of its budget on any
keyword.

\begin{figure}
  \begin{center}
    \scalebox{1.0}{
      \includegraphics{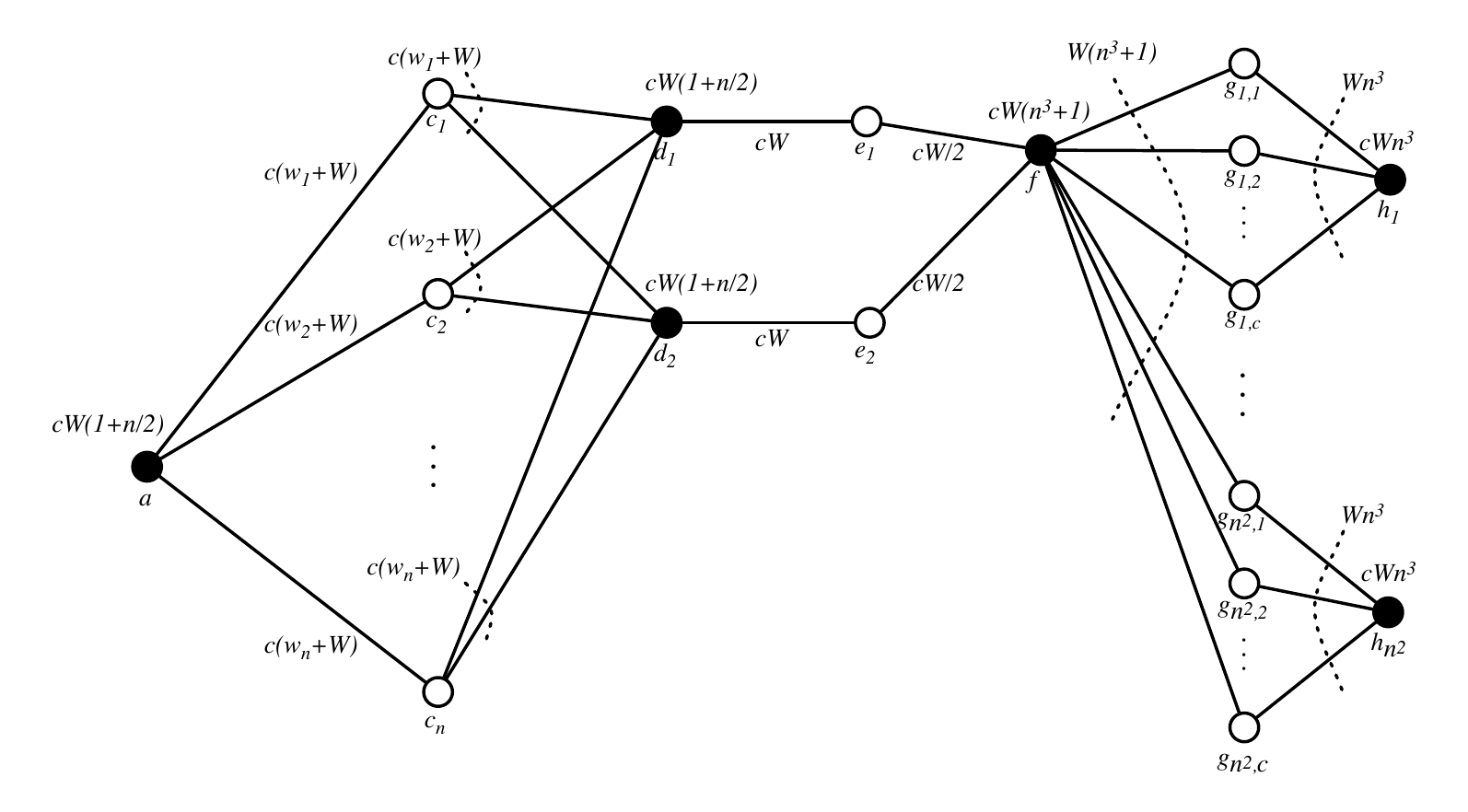}
    }
  \end{center}
  \caption{\small The 2PAA($c$) instance of the reduction.  Each bidder's
    budget is shown above its node, and the bids of bidders for
    keywords is shown near the corresponding edge.  }
  \label{fig:hardness_2PMBA_small}
\end{figure}

We first show that if the PARTITION instance is a ``yes'' instance,
then there exists a feasible solution to the 2PAA($c$) instance of
value at least $cW(n^5 + n + 2)$.  Let $S \subseteq \bracks{n}$ be
such that $|S| = n/2$ and $\sum_{i \in S} w_i = \sum_{i \in
  \overline{S}} w_i = W/2$.  We construct a solution to the 2PAA($c$)
instance as follows.  For every $i \in S$, allocate $c_i$ to $d_1$,
and for every $i \in \overline{S}$, allocate $c_i$ to $d_2$.  For each
of these allocations, choose $a$ as the second-price bidder.  This
will reduce the budget of $d_1$ and $d_2$ to exactly $cW/2$, and hence
the bids from $d_1$ to $e_1$ and from $d_2$ to $e_2$ will both be
reduced to $cW/2$.  Allocate $e_1$ to $f$ choosing $d_1$ as the
second-price bidder, and allocate $e_2$ to $f$ choosing $d_2$ as the
second-price bidder.  This will reduce the budget of $f$ to $cWn^3$.
The profit from the solution constructed so far is $cW(n+2)$.  Now
allocate $g_{1,1}, g_{1,2},\ldots,g_{1,c-1}$ to $f$, choosing $h_1$ as
the second-price bidder.  This will reduce the budget of $f$ to
$Wn^3$.  Hence, it can act as the second-price bidder for each of the
remaining keywords in $G$.  Allocate $g_{1,c}$ to $h_1$, choosing $f$
as the second-price bidder, and then, for $2 \leq i \leq n^2$ and $1
\leq k \leq c$, allocate $g_{i,k}$ to $h_i$, choosing $f$ as the
second-price bidder.  The profit obtained for each keyword in $G$ in
this assignment is $Wn^3$.  Since $|G| = cn^2$, the total profit of
the solution constructed is $cW(n+2) + cWn^5 = cW(n^5 + n + 2)$.

We now show that if there is a second-price matching in the 2PAA($c$)
instance of value at least $cW(n^3 + cn^2 + n + 2)$, then there must
be a partition of $w_1, \ldots, w_n$.  In such a matching, at most
$cW(n + 2)$ units of profit can be obtained from keywords $c_1,
\ldots, c_n, e_1, e_2$, since the initial second-highest bids on those
keywords sum to $cW(n+2)$.  Hence, at least $cW(n^3 + cn^2)$ profit
must come from the keywords in $G$.

Suppose that the budget of $f$ is greater than $cWn^3$ after keywords
$e_1$ and $e_2$ are allocated.  Note that at least $c$ of the keywords
in $G$ must be allocated to reach a profit of $cW(n^3 + cn^2)$ on these
keywords.  Consider what happens after the first $c$ of the keywords
in $G$ are assigned.  For each of these keywords, $f$ must have been
the first-price bidder, so its budget is reduced to an amount greater
than $0$ and less than or equal to $cW$.  Hence, for each keyword in
$G$ allocated henceforth, $f$ is the second-price bidder, and the
profit is at most $cW$.  Since there are at most $c(n^2 - 1)$ more
keywords in $G$, the total profit from the keywords in $G$ is at most
$cWn^3 + c^2W(n^2 - 1)$, which contradicts the fact that at least
$cW(n^3 + cn^2)$ units of profit must come from $G$.  Hence, we
conclude that the budget of $f$ is less than or equal to $cWn^3$ after
keywords $e_1$ and $e_2$ are allocated.

The budget of $f$ can only be smaller than $cWn^3$ if $f$ acts as the
first-price bidder for both $e_1$ and $e_2$.  But this can happen only
if the budgets of both $d_1$ and $d_2$ are reduced to an amount less
than or equal to $cW/2$.  For $j \in \braces{1,2}$, let $S_j \subseteq
\bracks{n}$ be the set of indices $i$ such that $d_j$ acts as the
first-price bidder for $i$.  For both $j$, we have that
\begin{equation}\label{eqn:impliespartition}
\sum_{i \in S_j} c(w_i + W) \geq \frac{cW}{2} + \frac{cWn}{2} \enspace .
\end{equation}
Rearranging (\ref{eqn:impliespartition}) yields $\sum_{i \in S_j} W
\geq W/2 + Wn/2 - \sum_{i \in S_j} w_i$, which implies $W |S_j| \geq
W/2 + Wn/2 - W$, and hence $|S_j| \geq n/2 - 1/2$.  By integrality,
then, $|S_j| \geq n/2$ for both $j$.  Hence $|S_j| = n/2$ for both
$j$, and using (\ref{eqn:impliespartition}) again, we have $\sum_{i
  \in S_j} c w_i + cW|S_j| \geq cW/2 + cWn/2$ which implies that
$\sum_{i \in S_j} w_i \geq W / 2$ for both $j$.  Therefore, the
partition defined by $S_1$ and $S_2$ is a solution to the PARTITION
instance.

To conclude the proof, note that there are $cn^2 + n + 2$ keywords in
the 2PAA($c$) instance.  Hence, if the PARTITION instance is a ``yes''
instance, then by (\ref{eqn:asymptotic1}), we can run an
$m/c'$-approximation algorithm to find a second-price matching of value
at least $cW(n^3 + cn^2 + n + 2)$, and if the PARTITION instance is
a ``no'' instance, then the value of the solution returned by such
an algorithm must be strictly less than $cW(n^3 + cn^2 + n + 2)$.
Hence, an $m/c'$-approximation algorithm for 2PAA($c$) can be
used to solve PARTITION.
\end{proof}

\begin{theorem}\label{thm:matching_2paa}
Let $c \geq 1$ be a constant integer. There is an $m/c$-approximation
to 2PAA($c$).
\end{theorem}
\begin{proof}
For each keyword $u \in U$, let $s_u$ be the second-highest bid for
$u$.  Consider the algorithm that selects the $c$ keywords with the
highest values of $s_u$ and then allocates these keywords to get $s_u$
for each of them (i.e., chooses the two highest bidders for $u$).
Since no bidder bids more than $1/c$ of its budget for any keyword, no
bids are reduced from their original values during this allocation.
Hence, the profit of this allocation is at least $(c/m) \sum_{u \in U}          
s_u$.  Since the value of the optimal solution cannot be larger than
$\sum_{u\in U} s_u$, it follows that this is an $m/c$-approximation to
2PAA($c$).
\end{proof}

%% file: offline2pm.tex
\section{Offline Second-Price Matching}\label{sec:offline2pm}

In this section, we turn our attention to the offline version of the
special case of Second-Price Matching (2PM).  Before we show our
bounds on the approximability of 2PM, we start with a simple proof
that it is NP-hard.  Then, in
Section~\ref{sec:hardness_offline2pm}, we show that 2PM is APX-hard,
and in Section~\ref{sec:approx_offline2pm}, we give a 2-approximation
for 2PM.
\begin{theorem}\label{thm:nphard}
  The Second-Price Matching Problem is NP-hard.
\end{theorem}
\noindent
Note that this result is subsumed by Theorem~\ref{thm:hardness_2pm}
below.  We present it anyway because it allows us to illustrate a
simpler reduction to the problem.
\begin{proof}
  We reduce from 3-SAT.  Given an instance of 3-SAT in which the
  variables are $X = \braces{x_1, \ldots, x_n}$ and the clauses are
  $\mathcal{C} = \braces{c_1, \ldots, c_k}$, we construct an instance
  of 2PM as follows:
  \begin{itemize}
    \item 
    For each variable $x_i \in X$, there is a keyword $v_i$.  We call
    these keywords the \emph{variable keywords}.  Each variable
    keyword $v_i$ is connected to two bidders $v_i^t$ and $v_i^f$.  We
    call these bidders the \emph{assignment bidders}.
  
    \item
    For each clause $c_j \in \mathcal{C}$, there is a keyword $u_j$
    and a bidder $b_j$.  We call these keywords and bidders
    \emph{clause keywords} and \emph{clause bidders}, respectively.
    Each clause keyword $u_j$ is connected to $b_j$ and three of the
    assignment bidders, one for each literal $\ell \in c_j$, chosen as
    follows.  If $\ell$ is of the form $x_i$ for some variable $x_i$,
    then $u_j$ is connected to $v_i^f$.  Otherwise, if $\ell$ is of
    the form $\overline{x_i}$ for some variable $x_i$, then $u_j$ is
    connected to $v_i^t$.
\end{itemize}
The keywords arrive in two phases: first the variable keywords and
then the clause keywords.  An example of this reduction is illustrated
in Figure~\ref{fig:nphardness}.

\begin{figure}
  \begin{center}
    \scalebox{0.9}{
      \includegraphics{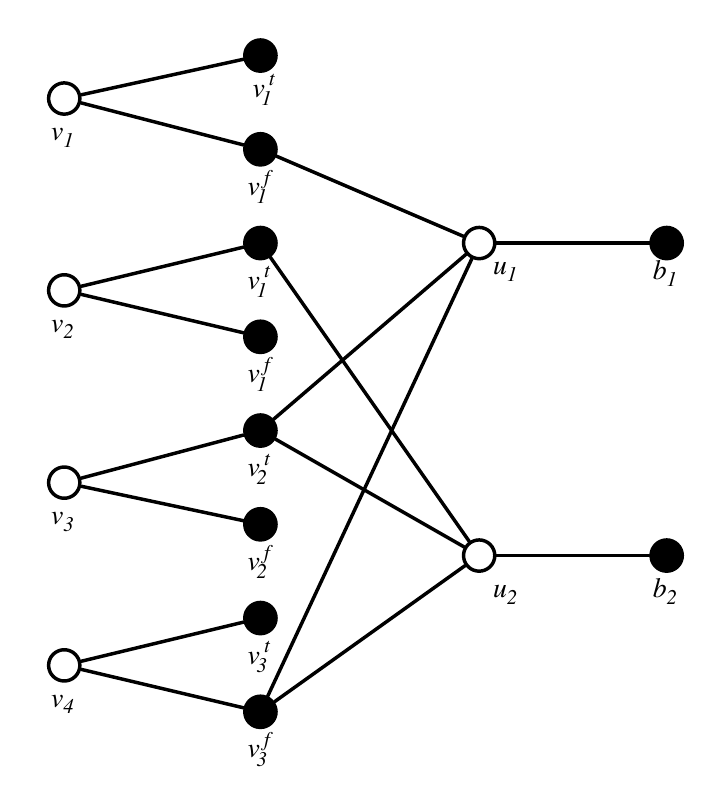}
    }
  \end{center}
  \caption{\small An example of the reduction from 3-SAT to 2PM.  The
    formula represented is $(x_1 \vee \overline{x_3} \vee x_4) \wedge
    (\overline{x_2} \vee \overline{x_3} \vee x_4)$.
  }\label{fig:nphardness}
\end{figure}

We now show that the 2PM instance has a second-price matching of value
$k+n$ if and only if there is a satisfying assignment to the 3-SAT
instance.  Suppose first that there is a satisfying assignment to the
3-SAT instance.  Let $h: X \rightarrow \braces{t,f}$ be the satisfying
assignment.  We construct a second-price matching as follows.  During
the first phase, assign each variable keyword $v_i$ to $v_i^{h(x_i)}$.
The profit from this phase is $n$.  During the second phase, assign
each clause keyword $u_j$ to $b_j$.  Since $c_j$ has at least one
satisfied literal $\ell$, the assignment bidder corresponding to
$\ell$ will not have been used in the first phase, and the profit for
assigning $u_j$ to $b_j$ is $1$.  Thus, the total profit from this
phase is $k$, and the total profit of the second-price matching is $k
+ n$.

On the other hand, if there is a second-price matching of size at
least $k + n$, then since there are $k + n$ keywords in the 2PMM
instance, the profit obtained from each keyword in the second-price
matching must be $1$.  This means that each variable keyword must have
been matched to one of its assignment bidders during the first phase.
Let $h: X \rightarrow \braces{t,f}$ be the assignment generated from
this matching, i.e., if $v_i$ was assigned to $v_i^\ell$ (for $\ell
\in \braces{t,f}$), then let $h(x_i) = \ell$.  Since the profit
obtained from each keyword is $1$, each clause keyword $u_j$ must have
been adjacent to at least two unused bidders when it was assigned,
including one of the assignment bidders, say $v_i^\ell$.  Hence,
$h(x_i) = \overline{\ell}$, and by construction of the 2PMM instance,
clause $c_j$ is satisfied by $h$.  We conclude that $h$ is a
satisfying assignment to the 3-SAT instance.
\end{proof}

\subsection{Hardness of Approximation}\label{sec:hardness_offline2pm}

To prove that 2PM is APX-hard, we reduce from vertex cover, using
the following result.
\begin{theorem}[Chleb\'{i}k and Chleb\'{i}kov\'{a} \cite{Chlebik06}]\label{thm:hardness_vc}
  It is NP-hard to approximate Vertex Cover on 4-regular graphs to
  within $53/52$.
\end{theorem}
\noindent
The precise statement of our hardness result is the following
theorem.
\begin{theorem}\label{thm:hardness_2pm}
  It is NP-hard to approximate 2PM to within a factor of $364/363$.
\end{theorem}
\begin{proof}
  Given a graph $G$ as input to Vertex Cover, we construct an instance
  $f(G)$ of 2PM as follows. First, for each edge $e \in E(G)$, we
  create a keyword with the same label (called an \emph{edge
    keyword}), and for each vertex $v \in V(G)$, we create a bidder
  with the same label (called a \emph{vertex bidder}). Bidder $v$ bids
  for keyword $e$ if vertex $v$ is one of the two end points of edge
  $e$.  (Recall that in 2PM, if a bidder makes a non-zero bid for a
  keyword, that bid is $1$.)  In addition, for each edge $e$, we
  create a unique bidder $x_e$ who also bids for $e$. Furthermore, for
  each vertex $v$, we create a gadget containing two keywords $h_v$
  and $l_v$ and two bidders $y_v$ and $z_v$. We let $v$ and $y_v$ bid
  for $h_v$; and $y_v$ and $z_v$ bid for $l_v$.  The keywords arrive
  in an order such that for each $v \in V(G)$, keyword $h_v$ comes
  before $l_v$, and the edge keywords arrive after all of the $h_v$'s
  and $l_v$'s have arrived.  An example of this reduction is shown in
  Figure~\ref{fig:hardness_2pm}.

  \begin{figure}
    \begin{center}
      \begin{tabular}{cc}
        \subfigure[]{\scalebox{0.9}{\label{fig:hardness_2pm:a}
            \includegraphics{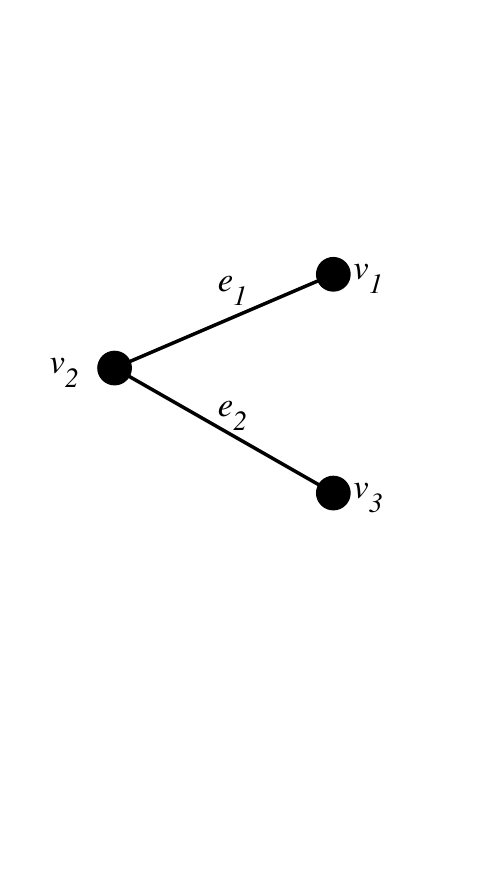}
          } 
        } &
        \subfigure[]{\scalebox{0.9}{\label{fig:hardness_2pm:b}
            \includegraphics{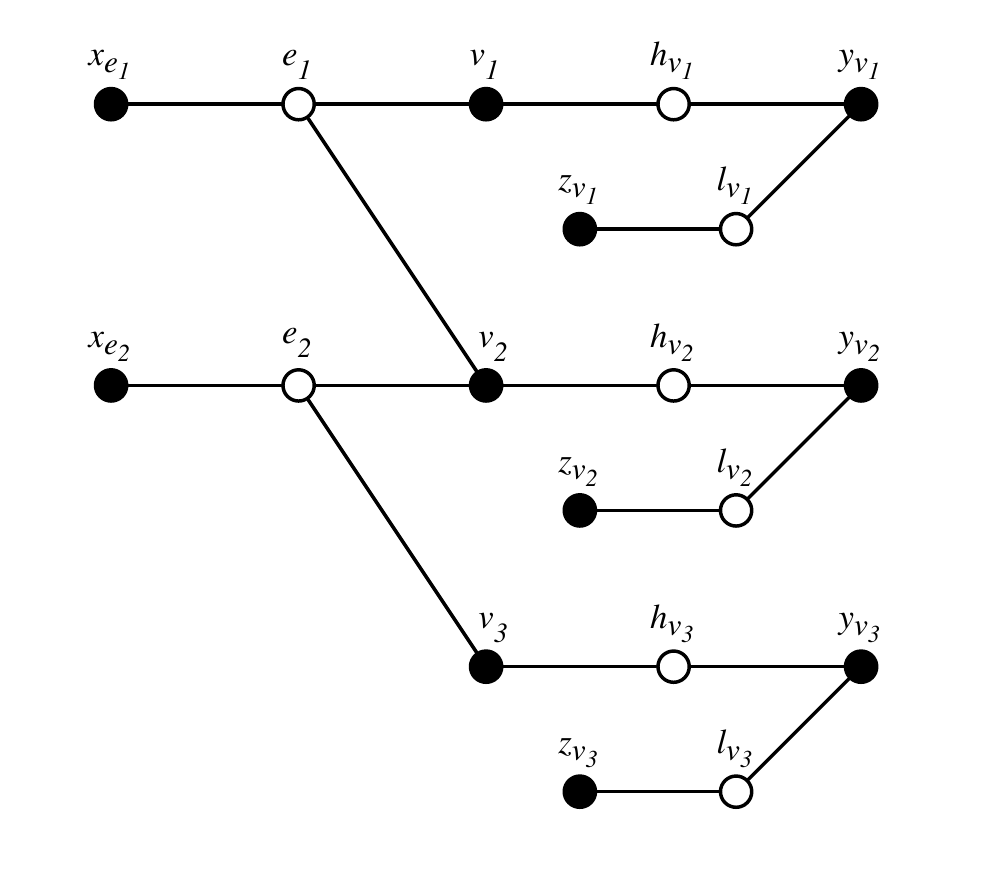}
          }
        }
      \end{tabular}
    \end{center}
    \caption{\small The reduction from an instance $G$ of vertex cover
      (Figure~\ref{fig:hardness_2pm:a}) to an instance $f(G)$ of
      2PM (Figure~\ref{fig:hardness_2pm:b}).
    } \label{fig:hardness_2pm}
  \end{figure}

  The following lemma provides the basis of the proof.
  \begin{lemma}\label{lem:hardness_2pm}
    Let $OPT_{VC}$ and $OPT_{2P}$ be the size of the minimum vertex
    cover of $G$ and the maximum second-price matching on $f(G)$,
    respectively. Then
    $$
      OPT_{2P} = 2|V(G)| + |E(G)| - OPT_{VC}
    $$
  \end{lemma}
\begin{proof}
  We first show that given a vertex cover $S$ of size $OPT_{VC}$ of
  $G$, we can construct a solution to the 2PM instance whose value is
  $2|V(G)| + |E(G)| - OPT_{VC}$. For each vertex $v \notin S$, we
  allocate $h_v$ to $v$ (with $y_v$ acting as the second-price bidder)
  and $l_v$ to $y_v$ (with $z_v$ acting as the second-price bidder),
  getting a profit of $2$ from the gadget for $v$.  For each vertex $v
  \in S$, we allocate $h_v$ to $y_v$ (with $v$ acting as the
  second-price bidder) and ignore $l_v$, getting a profit of $1$ from
  the gadget for $v$.  We then allocate each edge keyword $e$ to
  $x_e$.  During each of these edge keyword allocations, at least one
  of the two vertex bidders that bid for $e$ is still available, since
  $S$ is a vertex cover.  Hence, for each of these allocations, there
  is a bidder that can act as a second-price bidder, and the profit
  from the allocation is $1$.  This allocation yields a second-price
  matching of size $2|V(G)| + |E(G)| - OPT_{VC}$.  Therefore,
  $OPT_{2P} \geq 2|V(G)| + |E(G)| - OPT_{VC}$.
  
  To show that $OPT_{2P} \leq 2|V(G)| + |E(G)| - OPT_{VC}$, we start
  with an optimal solution to $f(G)$ of value $OPT_{2P}$ and
  construct a vertex cover of $G$ of size $2|V(G)| + |E(G)| -
  OPT_{2P}$.  To do this, we first claim that there exists an optimal
  solution of $f(G)$ in which every edge-keyword is allocated for a
  profit of $1$. Consider any optimal second-price matching of the
  instance.  Let $e$ be an edge-keyword $e$ that is not allocated for
  a profit of $1$.  If it is adjacent to a vertex bidder that is
  unassigned when $e$ arrives, then $e$ can be allocated to $x_e$ for
  a profit of $1$, which can only increase the value of the solution.
  Suppose, on the other hand, that both of its vertex bidders are not
  available when $e$ arrives.  Let $v$ be a vertex bidder that bids
  for $e$.  Since it is not available, $h_v$ must have been allocated
  to $v$.  We can transform this second-price matching to another one in
  which $h_v$ is assigned to $y_v$, $l_v$ is ignored and $e$ is
  assigned to $x_e$, with $v$ acting the second-price bidder in both
  cases.  This does not decrease the total profit of the solution.
  Hence, we can perform these transformations for each edge keyword
  $e$ that is not allocated for a profit of $1$ until we obtain a new
  optimal solution in which each edge keyword is allocated for a
  profit of $1$.

  Now consider an optimal second-price matching in which all edge
  keywords are allocated for a profit of $1$. Let $T \subseteq V(G)$
  be the set of vertices represented by vertex bidders that are not
  allocated any keywords in this second-price matching.  Then $|T| =
  2|V(G)| + |E(G)| - |OPT_{2P}|$, and $T$ is a vertex cover, which
  implies $OPT_{VC} \leq 2|V(G)| + |E(G)| - OPT_{2P}$. The lemma
  follows.
\end{proof}

Now, suppose that we have an $\alpha$-approximation for 2PM.  We will
show how to use this approximation algorithm and our reduction to
obtain an $((8\alpha - 7)/\alpha)$-approximation for Vertex Cover on
4-regular graphs.  By Theorem~\ref{thm:hardness_vc}, this means that
$(8\alpha - 7)/\alpha \geq 53/52$, and hence $\alpha \geq 364/363$,
unless $P = NP$.

To construct this $((8\alpha - 7)/\alpha)$-approximation algorithm,
given a 4-regular graph $G$, run the above reduction to obtain a 2PM
instance $f(G)$.  Then use the $\alpha$-approximation to obtain a
second-price matching $M$ whose value is at least $OPT_{2P}/\alpha$.
Now, just as in the proof of Lemma~\ref{lem:hardness_2pm}, we can
assume that in $M$, every edge keyword $e$ is allocated to $x_e$.
Hence, the set of vertices $T$ associated with the vertex bidders that
are not allocated a keyword form a vertex cover, and
\begin{eqnarray}
|T| 
& \leq & 2|V(G)| + |E(G)| - OPT_{2P}/\alpha \nonumber \\
& =    & 2|V(G)| + |E(G)| - (2|V(G)| + |E(G)| - OPT_{VC})/\alpha  \nonumber \\
& =    & (1-1/\alpha)(2|V(G)| + |E(G)|) + OPT_{VC}/\alpha \label{eqn:h2pm}
\end{eqnarray}
Since $G$ is $4$-regular, we have $OPT_{VC} \geq m/4 =
(2|V(G)|+|E(G)|)/8$, and hence by (\ref{eqn:h2pm}), we conclude that
$|T| \leq ((8\alpha - 7)/\alpha) OPT_{VC}$, which finishes the proof
of the theorem.
\end{proof}

\subsection{A 2-Approximation Algorithm}\label{sec:approx_offline2pm}

Consider an instance $G = (U \cup V,E)$ of the 2PM problem.  We
provide an algorithm that first finds a maximum matching $f : U
\rightarrow V$ and then uses $f$ to return a second-price matching
that contains at least half of the keywords matched by
$f$.\footnote{Note that $f$ is a partial function.}  Given a matching
$f$, call an edge $(u,v) \in E$ such that $f(u) \not= v$ an {\em
  up-edge} if $v$ is matched by $f$ and $f^{-1}(v)$ arrives before
$u$, and a {\em down-edge} otherwise.  Recall that we have assumed
without loss of generality that the degree of every keyword in $U$ is
at least two.  Therefore, every keyword $u \in U$ that is matched by
$f$ must have at least one up-edge or down-edge.
Theorem~\ref{thm:2_approx} 
%(proved in Appendix~\ref{appendix:2_approx}) 
shows that the following algorithm,
called ReverseMatch, is a $2$-approximation for 2PM.

\begin{center}
\small
\fbox{
\begin{tabular}{l}
{\bf ReverseMatch Algorithm:}\\
\hline
{\em Initialization:} \\
Find an arbitrary maximum matching $f:U\rightarrow V$ on $G$. \\
\hline
  {\em Constructing a 2nd-price matching:} \\
  Consider the matched keywords in reverse order of their arrival. \\
  For each keyword $u$: \\
  \hspace{.2in} If keyword $u$ is adjacent to a down-edge $(u,v)$: \\
    \hspace{.4in} Assign keyword $u$ to bidder $f(u)$ (with $v$ acting as the second-price bidder). \\
  \hspace{.2in} Else: \\
    \hspace{.4in} Choose an arbitrary bidder $v$ that is adjacent to keyword $u$. \\
    \hspace{.4in} Remove the edge $(f^{-1}(v),v)$ from $f$. \\
    \hspace{.4in} Assign keyword $u$ to bidder $f(u)$ (with $v$ acting as the second-price bidder). 
\end{tabular}
}
\end{center}
\normalsize

\begin{theorem}\label{thm:2_approx}
 The ReverseMatch algorithm is a 2-approximation.
\end{theorem}
\begin{proof}
Since the number of vertices matched by $f$ is an upper bound on the
profit of the maximum second-price matching on $G$, we need only to
prove that the second-price matching contains at least half of the
keywords matched by $f$.  By the behavior of the algorithm, it is
clear that whenever a vertex $u$ is matched to $f(u)$ in the
second-price matching, the profit obtained is $1$.  Furthermore, every
time an an edge is removed from $f$, a new keyword is added to the
second-price matching. Thus, the theorem follows.
\end{proof}

%% file: online2pm.tex
\section{Online Second-Price Matching}\label{sec:online2pm}

In this section, we consider the online 2PM problem, in which the
keywords arrive one-by-one and must be matched by the algorithm as
they arrive.  We start, in Section~\ref{sec:lower_online2pm}, by giving
a simple lower bound showing that no deterministic algorithm can
achieve a competitive ratio better than $m$, the number of keywords.
Then we move to randomized online algorithms and show that no
randomized algorithm can achieve a competitive ratio better than $2$.
In Section~\ref{sec:rand_upper_online2pm}, we provide a randomized
online algorithm that achieves a competitive ratio of
$2\sqrt{e}/(\sqrt{e}-1) \approx 5.083$.

\subsection{Lower Bounds}\label{sec:lower_online2pm}

The following theorem establishes our lower bound on deterministic
algorithms, which matches the trivial algorithm of arbitrarily
allocating the first keyword to arrive, and refusing to allocate any
of the remaining keywords.  
%The adversary for this lower bound offers
%the first keyword two bidders; whichever bidder the algorithm chooses
%for this first keyword will be needed as a second-price bidder for the
%rest of the keywords.  
\begin{theorem}\label{thm:deterministic_adversary}
For any $m$, there is an adversary that creates a graph with $m$
keywords that forces any deterministic algorithm to get a competitive
ratio no better than $m$.
\end{theorem}
\begin{proof}
The adversary shows the algorithm a single keyword (keyword $1$) that
has two adjacent bidders, $a_1$ and $b_2$.  If the algorithm does not
match keyword $1$ at all, a new keyword $2$ arrives that is adjacent
to two new bidders $a_2$ and $b_2$.  The adversary continues in this
way until either $m$ keywords arrive or the algorithm matches a
keyword $k < m$.  In the first case, the algorithm's performance is at
most $1$ (because it might match keyword $m$), whereas the adversary
can match all $m$ keywords.  Hence, the ratio is at least $m$.

In the second case, the adversary continues as follows.  Suppose
without loss of generality that the algorithm matches keyword $k$ to
$a_k$.  Then each keyword $i$, for $k+1 \leq i \leq m$, has one edge
to $a_k$ and one edge to a new bidder $c_i$.  Since the algorithm
cannot match any of these keywords for a profit, its performance is
$1$.  The adversary can clearly match each keyword $i$ for profit, for
$1 \leq i \leq k -1$, and if it matches keyword $k$ to $b_k$, then it
can use $a_k$ as a second-price bidder for the remaining keywords to
match them all to the $c_i$'s for profit.  Hence, the adversary can
construct a second-price matching of size at least $m$.
\end{proof}

We next show that no online (randomized) online algorithm for 2PM can
achieve a competitive ratio better than $2$.  
\begin{theorem}\label{thm:rand_lower_online2pm}
The competitive ratio of any randomized algorithm for 2PM must
be at least $2$.
\end{theorem}
\begin{proof}
We invoke Yao's Principle~\cite{Yao77} and construct a distribution of
inputs for which the best deterministic algorithm achieves an expected
performance of (asymptotically) $1/2$ the value of the optimal
solution.

Our distribution is constructed as follows.  The first keyword
arrives, and it is adjacent to two bidders.  Then the second keyword
arrives, and it is adjacent to one of the two bidders adjacent to the
first keyword, chosen uniformly at random, as well as a new bidder;
then the third keyword arrives, and it is adjacent to one of the
bidders adjacent to the second keyword, chosen uniformly at random, as
well as a new bidder; and so on, until the $m$-th keyword arrives.  We
call this a \emph{normal} instance.  To analyze the performance of the
online algorithms, we also define a \emph{restricted} instance to be
one that is exactly the same as a normal instance except that one of
the two bidders of the first keyword is marked \emph{unavailable},
i.e., he can not participate in any auction.
  
Clearly, an offline algorithm that knows the random choices beforehand
can allocate each keyword to the bidder that will not be adjacent to
the keyword that arrives next.  In this way, it can ensure that for
each keyword, there is a bidder that can act as a second-price bidder.
Hence for a normal instance, the optimal second-price matching obtains
a profit of $m$.

Consider the algorithm Greedy, which allocates a keyword to an
arbitrary adjacent bidder if and only if there is another available
bidder to act as a second-price bidder.  Our proof consists of two
steps: first, we will show that the expected performance of Greedy on
the normal instance is $(m + 1)/2$, and second we will prove that
Greedy is the best algorithm in expectation for both types of
instances.

Let $X_k^*$ and $Y_k^*$ be the \emph{expected} profit of Greedy on a
normal and a restricted instance of $k$ keywords, respectively (where
$X_0^*$ and $Y_0^*$ are both defined to be $0$). Given the first
keyword of a normal instance, Greedy allocates it to an arbitrary
bidder. Then, with probability $1/2$, it is faced with a normal
instance of $k-1$ keywords, and with probability $1/2$, it is faced
with a restricted instance of $k-1$ keywords. Therefore, for all
integers $k \geq 1$,
\begin{equation}
  \label{eq:1}
  X_{k}^* = 1/2(X_{k-1}^* + Y_{k-1}^*) + 1 \enspace .
\end{equation}
  
On the other hand, given the first keyword of a restricted instance,
Greedy just waits for the second keyword.  Then, with probability
$1/2$, the second keyword chooses the marked bidder, giving Greedy a
restricted instance of $k-1$ keywords, and with probability $1/2$, the
second keyword chooses the unmarked bidder, giving Greedy a normal
instance of $k-1$ keywords. Therefore, for all $k$,
\begin{equation}
  \label{eq:2}
  Y_k^* = 1/2(X^*_{k-1} + Y^*_{k-1}) \enspace .
\end{equation}
  
From (\ref{eq:1}) and (\ref{eq:2}) we have, for all $k$,
\begin{equation}
  \label{eq:3}
  Y_k^* = X_k^* - 1 \enspace .
\end{equation}
Plugging 
(\ref{eq:3}) for $k = m-1$ into (\ref{eq:1}) for $k = m$ yields
\begin{equation}
  \label{eq:4}
  X_m^* = X_{m-1}^* + 1/2 \enspace , 
\end{equation}
and hence, by induction $X_m^* = (m+1)/2$.

Now, we prove that Greedy is the best among all algorithms on these
two types of instances. In fact, we make it easier for the algorithms
by telling them beforehand how many keywords in the instance they will
need to solve. Let $X_m$ and $Y_m$ be the expected number of keywords
in the second-price matching produced by the \emph{best} algorithms
that ``know'' that they are solving a normal instance of size $m$ and
a restricted instance of size $m$, respectively. Let ${\cal A}_m$ and
${\cal B}_m$ denote these optimal algorithms.

We prove that $X_m \leq X_m^*$ and $Y_m \leq Y_m^*$ for all $m$ by
induction. The base case in which $m=1$ is easy, since no algorithm
can obtain a profit of more than one on a normal instance of one
keyword or more than zero on a restricted instance of one keyword.  We
now prove the induction step.

First, consider ${\cal A}_m$. When the first keyword arrives, ${\cal
  A}_m$ has two choices: either ignore it or allocate it to one of the
bidders. If ${\cal A}_m$ ignores the first keyword, its performance is
at most the performance of ${\cal A}_{m-1}$ on the remaining keywords,
which constitute a normal instance of $m-1$ keywords. On the other
hand, if ${\cal A}_m$ allocates the first keyword to one of the
bidders, then with probability $1/2$, it is faced with a normal
instance of $m-1$ keywords, and with probability $1/2$ it is faced
with a restricted instance of $m-1$ keywords.  The performance of
${\cal A}_m$ on these instance is at most the performance of ${\cal
  A}_{m-1}$ and ${\cal B}_{m-1}$, respectively. Thus, by the induction
hypothesis, (\ref{eq:3}), and (\ref{eq:4}), we have
\begin{eqnarray*}
  X_m & \leq & \max\{X_{m-1}, 1/2(X_{m-1} + Y_{m-1}) + 1\} \\
  & \leq & \max\{X_{m-1}^*, 1/2(X_{m-1}^* + Y_{m-1}^*) + 1\} \\
  & =    & \max\{X_{m-1}^*, 1/2(X_{m-1}^* + X_{m-1}^* - 1) + 1\} \\
  & =    & X_{m-1}^* + 1/2 \\
  & =    & X_m^* \enspace .
\end{eqnarray*}
  
Next, consider ${\cal B}_m$. When the first keyword arrives, ${\cal
  B}_m$ cannot allocate it for a profit.  If it allocates it for a
profit of $0$, then it is faced with a restricted instance of $m-1$
keywords.  If it does not allocate the keyword, then with probability
1/2, ${\cal B}_m$ is faced with a normal instance of $m-1$ keywords,
and with probability $1/2$, it is faced with a restricted instance of
$m-1$ keywords. Its performance on these instances is at most those of
${\cal A}_{m-1}$ and ${\cal B}_{m-1}$, respectively. Thus, by the
induction hypothesis and (\ref{eq:2}), we have
\begin{eqnarray*}
  Y_m & \leq & \max\{Y_{m-1},1/2(X_{m-1} + Y_{m-1})\} \\
  & \leq & \max \{Y_{m-1}^*,1/2(X_{m-1}^* + Y_{m-1}^*)\} \\
  & = & Y_m^* \enspace .
\end{eqnarray*}
This completes the proof.
\end{proof}

\subsection{A Randomized Competitive Algorithm}\label{sec:rand_upper_online2pm}

In this section, we provide an algorithm that achieves a competitive
ratio of $2\sqrt{e}/(\sqrt{e} - 1) \approx 5.083$.  The result builds
on a new generalization of the result that the Ranking algorithm for
online bipartite matching achieves a competitive ratio of $e/(e-1)
\approx 1.582$.  This was originally shown by Karp, Vazirani, and
Vazirani \cite{Karp90}, though a mistake was recently found in their
proof by Krohn and Varadarajan and corrected by Goel and Mehta
\cite{Goel08b}.

The online bipartite matching problem is merely the first-price
version of 2PM, i.e., the problem in which there is no requirement for
there to exist a second-price bidder to get a profit of $1$ for a
match.  The Ranking algorithm chooses a random permutation on the
bidders $V$ and uses that to choose matches for the keywords $U$
as they arrive.  This is described more precisely below.

\begin{center}
\small
\fbox{
\begin{tabular}{l}
{\bf Ranking Algorithm:}\\
\hline
\emph{Initialization}: \\
Choose a random permutation (ranking) $\sigma$ of the bidders $V$.\\
\hline
  \emph{Online Matching}: \\
  Upon arrival of keyword $u \in U$: \\
  \hspace{.2in} Let $N(u)$ be the set of neighbors of $u$ that have not been matched yet. \\
  \hspace{.2in} If $N(u) \neq \emptyset$, match $u$ to the bidder $v \in N(u)$ that minimizes $\sigma(v)$.
\end{tabular}
}
\end{center}
\normalsize

\noindent
Karp, Vazirani, and Vazirani, and Goel and Mehta prove the following
result.
\begin{theorem}[Karp, Vazirani, and Vazirani \cite{Karp90}
    and Goel and Mehta \cite{Goel08b}]\label{thm:old_ranking} 
  The Ranking algorithm for online bipartite matching achieves a
  competitive ratio of $e/(e-1) + o(1)$.
\end{theorem}
\noindent
In order to state our generalization of this result, we define
the notion of a $\emph{left $k$-copy}$ of a bipartite graph $G = (U
\cup V, E)$.  Intuitively, a left $k$-copy of $G$ makes $k$ copies of
each keyword $u \in U$ such that the neighborhood of a copy of $u$ is
the same as the neighborhood of $u$.  More precisely, we have the
following definition.
\begin{definition}\label{def:kcopy}
Given a bipartite graph $G = (U_G \cup V, E_G)$, a {\bf left $k$-copy}
of $G$ is a graph $H = (U_H \cup V,E_H)$ for which $|U_H| = k|U_G|$
and for which there exists a map $\zeta: U_H \rightarrow U_G$ such
that
\begin{itemize}
\item
for each $u_G \in U_G$ there are exactly $k$ vertices $u_H \in U_H$
such that $\zeta(u_H) = u_G$, and

\item
for all $u_H \in U_H$ and $v \in V$, $(u_H, v) \in E_H$ if and
only if $(\zeta(u_H), v) \in E_G$.
\end{itemize}
\end{definition}

\noindent
Our generalization of Theorem~\ref{thm:old_ranking} describes the
competitive ratio of Ranking on a graph $H$ that is a left $k$-copy of
$G$.  Its proof, presented in Appendix~\ref{appendix:ranking}, builds
on the proof of Theorem~\ref{thm:old_ranking} presented by Birnbaum
and Mathieu~\cite{Birnbaum08}.
\begin{theorem}\label{thm:ranking}
Let $G = (U_G \cup V, E_G)$ be a bipartite graph that has a maximum
matching of size $OPT_{1P}$, and let $H = (U_H \cup V, E_H)$ be a left
$k$-copy of $G$.  Then the expected size of the matching returned by
Ranking on $H$ is at least
\begin{equation*}
k OPT_{1P} \paren{1 - \frac{1}{e^{1/k}} + o(1)} \enspace .
\end{equation*}
\end{theorem}

\noindent
Using this result, we are able to prove that the following algorithm,
called RankingSimulate, achieves a competitive ratio of
$2\sqrt{e}/(\sqrt{e}-1)$.
\begin{center}
\small
\fbox{
\begin{tabular}{l}
{\bf RankingSimulate Algorithm:}\\
\hline
\emph{Initialization:} \\
Set $M$, the set of \emph{matched} bidders, to $\emptyset$. \\
Set $R$, the set of \emph{reserved} bidders, to $\emptyset$. \\
Choose a random permutation (ranking) $\sigma$ of the bidders $V$.\\
\hline
  \emph{Online Matching:} \\
  Upon arrival of keyword $u \in U$: \\  
  \hspace{.2in} Let $N(u)$ be the set of neighbors of 
  $u$ that are not in $M$ or $R$. \\
  \hspace{.2in} If $N(u) = \emptyset$, do nothing. \\
  \hspace{.2in} If $|N(u)| = 1$, let $v$ be the single bidder in $N(u)$. \\
  \hspace{.4in} With probability $1/2$, match $u$ to $v$ and add $v$
  to $M$, and \\
  \hspace{.4in} With probability $1/2$, add $v$ to $R$. \\
  \hspace{.2in} If $|N(u)| \geq 2$, let $v_1$ and $v_2$ be the two
  distinct bidders in $N(u)$ that minimize $\sigma(v)$. \\
  \hspace{.4in} With probability $1/2$, match $u$ to $v_1$, add
  $v_1$ to $M$, and add $v_2$ to $R$, and \\
  \hspace{.4in} With probability $1/2$, match $u$ to $v_2$, add
  $v_1$ to $R$, and add $v_2$ to $M$.
\end{tabular}
}
\end{center}
\normalsize

Let $G = (U_G \cup V, E_G)$ be the bipartite input graph to 2PM, and
let $H = (U_H \cup V, E_H)$ be a left $2$-copy of $H$.  In the arrival
order for $H$, the two copies of each keyword $u_G \in U$ arrive in
sequential order.  We start with the following lemma.

\begin{lemma}\label{lem:relateranking}
Fix a ranking $\sigma$ on $V$.  For each bidder $v \in V$, let $X_v$
be the indicator variable for the event that $v$ is matched by Ranking
on $H$, when the ranking is $\sigma$.\footnote{Note that once $\sigma$
  is fixed, $X_v$ is deterministic.}  Let $X'_v$ be the indicator
variable for the event that $v$ is matched by RankingSimulate on $G$,
when the ranking is $\sigma$.  Then $\expect(X'_v) = X_v / 2$.
\end{lemma}
\begin{proof}
It is easy to establish the invariant that for all $v \in V$, $X_v =
1$ if and only if RankingSimulate puts $v$ in either $M$ or $R$.
Furthermore, each bidder $v \in V$ is put in $M$ or $R$ at most once
by RankingSimulate.  The lemma follows because each time
RankingSimulate adds a bidder $v$ to $M$ or $R$, it matches it with
probability $1/2$.
\end{proof}

With Theorem~\ref{thm:ranking} and Lemma~\ref{lem:relateranking}, we
can now prove the main result of this section.
\begin{theorem}
The competitive ratio of RankingSimulate is $2\sqrt{e}/(\sqrt{e} - 1)
\approx 5.083$.
\end{theorem}
\begin{proof}
For a permutation $\sigma$ on $V$, let $\rankingsim(\sigma)$ be the
matching of $G$ returned by RankingSimulate, and let
$\ranking(\sigma)$ be the matching of $H$ returned by Ranking.
Lemma~\ref{lem:relateranking} implies that, conditioned on $\sigma$,
$\expect(|\rankingsim(\sigma)|) = |\ranking(\sigma)| / 2$.  
By
Theorem~\ref{thm:ranking},
\begin{equation*}
\expect (|\rankingsim(\sigma)|)
= \frac{1}{2} \expect (|\ranking(\sigma)|)
\geq  OPT_{1P} \paren{1 - 1/e^{1/2} + o(1)} \enspace .
\end{equation*}

Fix a bidder $v \in V$.  Let $P_v$ be the profit from $v$ obtained by
RankingSimulate.  Suppose that $v$ is matched by RankingSimulate to
keyword $u \in U_G$.  Recall that we have assumed without loss of
generality that the degree of $u$ is at least $2$.  Let $v' \ne v$ be
another bidder adjacent to $u$.  Then, given that $v$ is matched to
$u$, the probability that $v'$ is matched to any keyword is no greater
than $1/2$.  Therefore, $\expect(P_v|\mbox{$v$ matched}) \geq 1/2$.
Hence, the expected value of the second-price matching returned by
RankingSimulate is
\begin{eqnarray*}
\sum_{v \in V} \expect(P_v) & = &
\sum_{v \in V} \expect(P_v | \mbox{$v$ matched}) \pr(\mbox{$v$ matched}) \\
& \geq & \frac12 \sum_{v \in V} \pr(\mbox{$v$ matched}) \\
& =    & \frac12 \expect(|\rankingsim(\sigma)|) \\
& \geq & \frac12 OPT_{1P} \paren{1 - 1/e^{1/2} + o(1)} \\
& \geq & \frac12 OPT_{2P} \paren{1 - 1/e^{1/2} + o(1)} \enspace ,
\end{eqnarray*}
where $OPT_{2P}$ is the size of the optimal second-price matching on $G$.
\end{proof}

%% file: conclusion.tex
\section{Conclusion}

In this paper, we have shown that the complexity of the Second-Price
Ad Auctions problem is quite different from that of the more studied
First-Price Ad Auctions problem, and that this discrepancy extends to
the special case of 2PM, whose first-price analogue is bipartite
matching.  On the positive side, we have given a 2-approximation
for offline 2PM and a 5.083-competitive algorithm for online 2PM.

Some open questions remain.  Closing the gap between $2$ and $364/363$
in the approximability of offline 2PM is one clear direction for
future research, as is closing the gap between $2$ and $5.083$ in the
competitive ratio for online 2PM.  Another question we leave open is
whether the analysis for RankingSimulate is tight, though we expect
that it is not.  

%A final direction for future research is to consider
%the multiple-slot version of 2PAA and 2PM.

%% file: appendix_models.tex
\section{Discussion of Related Models}\label{appendix:models}
In this section, we discuss the relationship between the ``strict''
and ``non-strict'' models of Goel et al.~\cite{Goel08a} and our model.
In the strict model, a bidder's bid can be above his remaining budget,
as long as the remaining budget is strictly positive.  In the
non-strict model, bidders can keep their bids positive even after
their budget is depleted.  In both models a bidder is not charged more
than his remaining budget for a slot.  Therefore, in the non-strict
model, if a bidder is allocated a slot after his budget is fully
depleted, then he gets the slot for free.

Given an instance $A$, let $OPT_{2P}$ be the optimal solution value in
our model; let $OPT_{strict}$ be the optimal solution value under the
strict model; and let $OPT_{non-strict}$ be the optimal solution value
under the non-strict model.  Surprisingly, even though the strict and
non-strict models seem more permissive, it is possible for $OPT_{2P}$
to be $\Omega(m)$ times as big as $OPT_{strict}$ and $OPT_{non-strict}$, 
even when $R_{min}$ is a large constant $c$.  This is shown in
Figure~\ref{fig:higher_revenue}.

\begin{figure}
  \begin{center}
    \scalebox{0.8}{\includegraphics{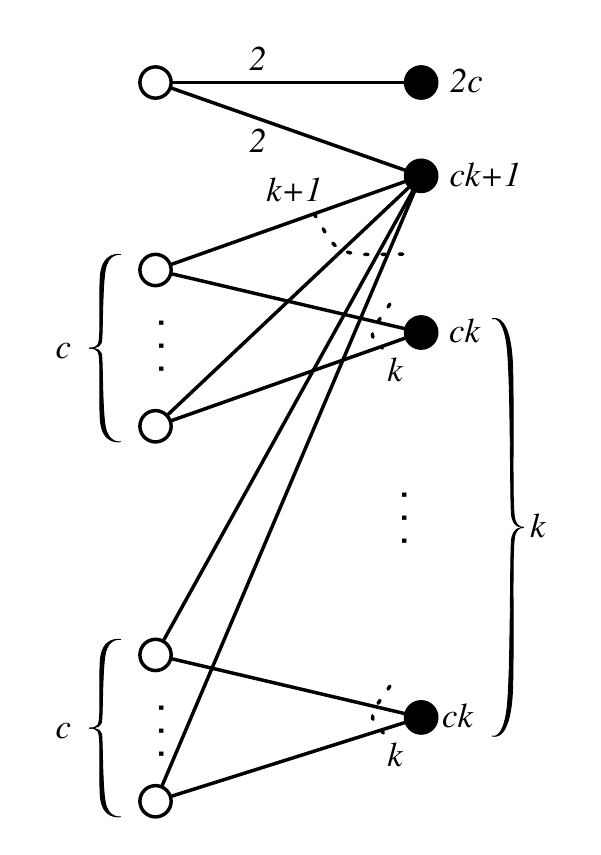}}
    \caption{
      \small In this example, $R_{min}$ is equal to a constant $c$,
      i.e., every bid is at most $1/c$ of the budget of the bidder.  In
      the strict model, all keywords except the first must be allocated
      to the second bidder at a price of $k$ (or the remaining budget if
      it's smaller).  Thus, the total profit on this input for the
      strict model is at most $ck+3$.  On the other hand, in our model,
      if we allocate the first keyword to the second bidder, and then
      the next $c-1$ keywords to the second bidder, that bidder's budget
      is reduced to $k-1$. Thus, all of the remaining keywords can be
      allocated to the lower bidder at a price of $k-1$, for a total
      revenue exceeding $ck(k-1)$. For $k$ large, this ratio is $\Omega
      (k) = \Omega(m)$.  }\label{fig:higher_revenue}
  \end{center}
\end{figure}

On the other hand, we show below that the optimal values of the two
models of Goel et al.\ cannot be better than the optimal value of our
model by more than a constant factor.  
\begin{theorem}\label{thm:relationship}
For any instance $A$, $OPT_{non-strict} \leq (2+1/R_{min}) OPT_{strict} 
\leq 8 (2 + 1/R_{min}) OPT_{2P}$.
\end{theorem}
The first inequality is proved by Goel et al.~\cite{Goel08a}, so we
must only prove that $OPT_{strict} \leq 8 OPT_{2P}$.

The core of our argument is a reduction from 2PAA to the First-Price
Ad Auctions problem (1PAA),\footnote{Recall that this problem has also
  been called the {\em Adwords} problem~\cite{Mehta07} and the {\em
    Maximum Budgeted Allocation}
  problem~\cite{Azar08,Chakrabarty08,Srinivasan08}.} in which only one
bidder is chosen for each keyword and that bidder pays the minimum of
its bid and its remaining budget.  Given an instance $A$ of 2PAA, we
construct an instance $A'$ of 1PAA problem by replacing each bid
$b_{u,v}$ by
\begin{equation*}
  b'_{u,v} \triangleq \max_{v' \ne v~:~ b_{u,v'} \leq b_{u,v}} b_{u,v'}
  \enspace .
\end{equation*}
Denote by $OPT_{1P}(A')$ the optimal value of the first-price model on
$A'$. The following two lemmas prove Theorem~\ref{thm:relationship} by
relating both $OPT_{non-strict}(A)$ and $OPT_{2P}(A)$ to
$OPT_{1P}(A')$.

\begin{lemma}
 $OPT_{non-strict}(A) \leq OPT_{1P}(A')$.
\end{lemma}
\begin{proof}
  For an instance $A$, we can view a non-strict second-price
  allocation of $A$ as a pair of (partial) functions $f_1$ and $f_2$
  from the keywords $U$ to the bidders $V$, where $f_1$ maps each
  keyword to the bidder to which it is allocated and $f_2$ maps each
  keyword to the bidder acting as its second-price bidder.  Thus, if
  $f_1(u) = v$ and $f_2(u) = v'$ then $u$ is allocated to $v$ and the
  price $v$ pays is $b_{u,v'}$.  We have, for all such $u$, $v$, and
  $v'$, that $b_{u,v'} \leq b'_{u,v}$.

  We construct the first-price allocation on $A'$ defined by $f_1$ and
  claim that the value of this first-price allocation is at least the
  value of the non-strict allocation defined by $f_1$ and $f_2$. It
  suffices to show that for any bidder $v$, the profit that the
  non-strict allocation gets from $v$ is at most the profit that the
  first-price allocation gets from $v$, or in other words,
  $$ 
  \min\left(B_v, \sum_{u: f_1(u) = v}
  b_{u,f_2(u)}\right) \leq \min\left(B_v, \sum_{u: f_1(u)
    = v} b'_{u,v}\right) \enspace .
  $$ 
  This inequality follows trivially from the fact that
  $b_{u,f_2(u)} \leq b'_{u,f_1(u)}$ for all allocated keywords $u$,
  and hence the lemma follows.
\end{proof}

\begin{lemma}
  $OPT_{1P}(A') \leq 8OPT_{2P}(A)$.
\end{lemma}
\begin{proof}
  Given an optimal first-price allocation of $A'$, we can assume
  without loss of generality that each bidder's budget can only be
  exhausted by the last keyword allocated to it, or, more formally, if
  $u_1, u_2, \ldots u_k$ are the keywords that are allocated to a
  bidder $v$ and they come in that order, then we can assume that
  $\sum_{i=1}^{k-1} b'_{u_i,v} < B_v$.  The reason for this is that if
  for some $j < k$, $\sum_{i=1}^{j-1} b'_{u_i,v} < B_v$ and
  $\sum_{i=1}^j b'_{u_i,v} \geq B_v$, then we can ignore the allocation
  of $u_{j+1}, \ldots u_k$ to $v$ without losing any profit.
 
  With this assumption, we design a randomized algorithm that
  constructs a second-price allocation on $A$ whose expected value in
  our model is at least $1/8$ of the first-price allocation's value.
  Viewing the first-price allocation of $A'$ as a (partial) function
  $f$ from the keywords $U$ to the bidders $V$ and denoting by
  $s(u,v)$ the bidder $v'$ for which $b_{v'u} = b'_{vu}$, the
  algorithm is as follows.

\begin{center}
\small
\fbox{
\begin{tabular}{l}
{\bf Random Construction:}\\
\hline
  Randomly mark each bidder with probability $1/2$. \\
  For each unmarked bidder $v$: \\
    \hspace{.2in} Let $S_v = \emptyset$.\\
    \hspace{.2in} For each keyword $u$ such that $f(u) = v$:\\
      \hspace{.4in} If $s(u,v)$ is marked: $S_v = S_v \cup \{u\}$.\\
    \hspace{.2in} Assume that $S_v = \{u_1, u_2, \ldots u_k\}$, where $u_1, u_2, \ldots u_k$ come in that order.\\
      \hspace{.4in} If $\sum_{i=1}^k b'_{u_i,v} \leq B_v$:\\
        \hspace{.6in} Let $f_1(u_i) = v$ and $f_2(u_i) = s(u_i,v)$ for all $i \leq k$.\\
      \hspace{.4in} Else: \\
        \hspace{.6in} If $\sum_{i=1}^{k-1} b'_{u_i,v} \geq b'_{u_k,v}$: 
          let $f_1(u_i) = v$ and $f_2(u_i) = s(u_i,v)$ for all $i \leq k-1$. \\
        \hspace{.6in} Else: let $f_1(u_k) = v$ and $f_2(u_k) = s(u_k, v)$.\\
\end{tabular}
}
\end{center}
\normalsize

We claim that for the $f_1$ and $f_2$ defined by this construction,
whenever $f(u_i)$ is set to $v$, the profit from that allocation is
$b'_{u_i,v}$.  This is not trivial because in our model, if a bidder's
remaining budget is smaller than its bid for a keyword, it changes its
bid for that keyword to its remaining budget. However, one can easily
verify that in all cases, if we set $f_1(u_i) = v$ and $f_2(u_i) =
s(u_i, v)$, the remaining budget of $v$ is at least $b'_{u_i,v} =
b_{u_i,s(u_i,v)}$. Thus, the (modified) bid of $f_1(u_i)$ for $u_i$ is
still at least the original bid of $f_2(u_i)$ for $u_i$.

We claim that the expected value of the second-price allocation
defined by $f_1$ and $f_2$ is at least $1/8 OPT_{1P}(A')$. For each
bidder $v$, let $X_v$ be the random variable denoting the profit that
$f_1$ and $f_2$ get from $v$, and let $Y_v$ be the profit that $f$
gets from $v$. We have $OPT_{1P}(A') = \sum_v Y_v$, so it suffices to
show that $E(X_v) \geq 1/8Y_v$ for all $v \in V$.

Consider any $v \in V$ that is unmarked. Let $T_v = \{u: f(u) =
v\}$. If $\sum_{u \in S_v} b'_{u,v} \leq B_v$ then $X_v = \sum_{u \in
  S_v} b'_{u,v}$. If $\sum_{u \in S_v} b'_{u,v} > B_v$ then $X_v \geq
\sum_{u \in S_v} b'_{u,v}/2$. Thus, in both case, we have
$$
  E[X_v|v\ \textrm{is unmarked}] \geq E[\sum_{u \in S_v} b'_{u,v}/2 | v\ \textrm{is unmarked}] = \sum_{u \in T_v} b'_{u,v}/4 = Y_v/4 \enspace ,
$$
which implies
$$
  E[X_v] \geq E[X_v|v\ \textrm{is unmarked}]Pr[v\ \textrm{is unmarked}] = 1/2\cdot Y_v/4 = Y_v/8 \enspace .
$$
\end{proof}

%% file: appendix_kvv.tex
\section{Proof of Theorem~\ref{thm:ranking}}\label{appendix:ranking}
In this appendix, we provide a full proof of
Theorem~\ref{thm:ranking}.  The proof presented here is quite similar
to the simplified proof of Theorem~\ref{thm:old_ranking} presented by
Birnbaum and Mathieu~\cite{Birnbaum08}.  For intuition into the proof
presented here, the interested reader is referred to that
work.\footnote{For those familiar with the proof in~\cite{Birnbaum08},
  the main difference between the proof of Theorem~\ref{thm:ranking}
  presented here and the proof of Theorem~\ref{thm:old_ranking}
  presented in \cite{Birnbaum08} appears in
  Lemma~\ref{lem:rankingmain}.  Instead of letting $u$ be the single
  vertex that is matched to $v$ by the perfect matching, as is done in
  \cite{Birnbaum08}, we choose $u$ uniformly at random from one of the
  $k$ vertices that correspond to the vertex that is matched to $v$ by
  the perfect matching.  The rest of the proof is essentially the
  same, but we present its entirety here for completeness.}

Let $G = (U_G \cup V, E_G)$ be a bipartite graph and let $H = (U_H
\cup V, E_H)$ be a left $k$-copy of $G$. Let $\zeta : U_H \rightarrow
U_G$ be a map that satisfies the conditions of
Definition~\ref{def:kcopy}.  Let $M_G \subseteq E_G$ be a maximum
matching of $G$.

Let $\ranking(H, \pi, \sigma)$ denote the matching constructed on $H$
for arrival order $\pi$, when the ranking is $\sigma$.  Consider
another process in which the vertices in $V$ arrive in the order given
by $\sigma$ and are matched to the available vertex $u \in U_H$ that
minimizes $\pi(u)$.  Call the matching constructed by this process
$\ranking'(H, \pi, \sigma)$.  
It is not hard to see that these matchings are identical, a fact
that is proved in \cite{Karp90}.
\begin{lemma}[Karp, Vazirani, and Vazirani \cite{Karp90}]\label{lem:duality}
For any permutations $\pi$ and $\sigma$, $\ranking(H, \pi, \sigma) =
\ranking'(H, \pi,\sigma)$.\end{lemma}

The following monotonicity lemma shows that removing vertices in $H$
can only decrease the size of the matching returned by Ranking.

\begin{lemma}\label{lem:monotonicity}
Let $\pi_H$ be an arrival order for the vertices in $U_H$, and let
$\sigma_H$ be a ranking on the vertices in $V$.  Suppose that $x$ is a
vertex in $U_H \cup V$, and let $H' = (U_{H'}, V_{H'}, E_{H'}) = H
\setminus \braces{x}$.  Let $\pi_{H'}$ and $\sigma_{H'}$ be the
orderings of $U_{H'}$ and $V_{H'}$ induced by $\pi_H$ and $\sigma_H$,
respectively.  Then $|\ranking(H', \pi_{H'}, \sigma_{H'})| \leq
|\ranking(H, \pi_H, \sigma_H)|$.
\end{lemma}
\begin{proof}
Suppose first that $x \in U_H$.  In this case, $V = V_{H'}$ and
$\sigma_{H} = \sigma_{H'}$.  Let $Q_t(H) \subseteq V$ be the set of
vertices matched to vertices in $U_H$ that arrive at or before time
$t$ (under arrival order $\pi_H$ and ranking $\sigma_H$), and let
$Q_t(H') \subseteq V$ be the set of vertices matched to vertices in
$U_{H'}$ that arrive at or before time $t$ (under arrival order
$\pi_{H'}$ and ranking $\sigma_H$).  We prove by induction on $t$ that
$Q_{t-1}(H') \subseteq Q_{t}(H)$, which by substituting $t = n$ is
sufficient to prove the claim.  The statement holds when $t = 1$,
since $Q_0(H') = \emptyset$.  Now supposing we have $Q_{t-2}(H')
\subseteq Q_{t-1}(H)$, we prove $Q_{t-1}(H') \subseteq Q_t(H)$.
Suppose that $t$ is at or before the time that $x$ arrives in $\pi_H$.
Then clearly $Q_{t-1}(H') = Q_{t-1}(H) \subseteq Q_{t}(H)$.  Now
suppose that $t$ is after the time that $x$ arrives in $\pi_H$.  Let
$u$ be the vertex that arrives at time $t-1$ in $\pi_{H'}$.  If $u$ is
not matched by $\ranking(H',\pi_{H'},\sigma_H)$, then $Q_{t-1}(H') =
Q_{t-2}(H') \subseteq Q_{t-1}(H) \subseteq Q_{t}(H)$.  Now suppose
that $u$ is matched by $\ranking(H', \pi_{H'}, \sigma_{H})$, say to
vertex $v'$. We show that $v' \in Q_t(H)$, which by the induction
hypothesis, is enough to prove that $Q_{t-1}(H') \subseteq Q_{t}(H)$.
Note that $u$ arrives at time $t$ in $\pi_H$.  Let $v$ be the vertex
to which $u$ is matched by $\ranking(H, \pi_H, \sigma_H)$.  If $v =
v'$, we are done, so suppose that $v \not= v'$.  Since $v \not\in
Q_{t-1}(H)$, it follows by the induction hypothesis that $v \not\in
Q_{t-2}(H')$.  Therefore, vertex $v$ is available to be matched to $u$
when it arrives in $\pi_{H'}$.  Since $\ranking(H', \pi_{H'},
\sigma_{H})$ matched $u$ to $v'$ instead, $v'$ must have a lower rank
than $v$ in $\sigma_H$.  Since $\ranking(H, \pi_H, \sigma_H)$ chose
$v$, vertex $v'$ must have already been matched when vertex $u$
arrived at time $t$ in $\pi_H$, or, in other words, $v' \in Q_{t-1}(H)
\subseteq Q_t(H)$.

Now suppose that $x \in V$.  In this case, $U_{H} = U_{H'}$ and $\pi_H
= \pi_{H'}$.  Let $R_t(H) \subseteq U_H$ be the set of vertices
matched to vertices in $V$ that are ranked less than or equal to $t$
(under arrival order $\pi_H$ and ranking $\sigma_H$), and let $R_t(H')
\subseteq U_H$ be the set of vertices matched to vertices in $V$ that
are ranked less than or equal to $t$ (under arrival order $\pi_{H}$
and ranking $\sigma_{H'}$).  Then by Lemma~\ref{lem:duality}, we can
apply the same argument as before to show that $R_{t-1}(H') \subseteq
R_t(H)$ for all $t$, which by substituting $t = n$, is sufficient to
prove the claim.
\end{proof}

We define the following notation.  For all $u_G \in U_G$, let
$\zeta^{-1}(u_G)$ be the set of all $u_H \in U_H$ such that
$\zeta(u_H) = u_G$, and for any subset $U_G' \subseteq U_G$, let
$\zeta^{-1}(U_G')$ be the set of all $u_H \in U_H$ such that
$\zeta(u_H) \in U_G'$.  The following lemma shows that we can assume
without loss of generality that $M_G$ is a perfect matching.

\begin{lemma}\label{lem:perfectok}
Let $U' \subseteq U_G$ and $V' \subseteq V$ be the subset of vertices
that are in $M_G$.  Let $G'$ be the subgraph of $G$ induced by $U'
\cup V'$, and let $H'$ be the subgraph of $H$ induced by
$\zeta^{-1}(U') \cup V'$.  Then the expected size of the matching
produced by Ranking on $H'$ is no greater than the expected size of
the matching produced by Ranking on $H$.
\end{lemma}
\begin{proof}
The proof follows by repeated application of
Lemma~\ref{lem:monotonicity} for all $x$ that are not in
$\zeta^{-1}(U') \cup V'$.
\end{proof}

In light of Lemma~\ref{lem:perfectok}, to prove
Theorem~\ref{thm:ranking}, it is sufficient to show that the expected
size of the matching produced by Ranking on $H'$ is at least
$(1-1/e^{1/k} - o(1))|M_G|$.  To simplify notation, we instead assume
without loss of generality that $G = G'$, and hence $G$ has a perfect
matching.  Let $n = OPT_{1P} = |M_G| = |V|$.  Henceforth, fix an arrival order
$\pi$.  To simplify notation, we write $\ranking(\sigma)$ to mean
$\ranking(H, \pi, \sigma)$.

Let $f : U_H \rightarrow V$ be a map such that for all $v \in V$,
there are exactly $k$ vertices $u \in U_H$ such that $f(u) = v$.  The
existence of such a map $f$ follows from the assumption that $G$ has a
perfect matching.  For any vertex $v \in V$ let $f^{-1}(v)$ be the set
of $u \in U_H$ such that $f(u) = v$.  We proceed with the following
two lemmas.

\begin{lemma}\label{lem:easy}
Let $u \in U_H$, and let $v = f(u)$.  For any ranking $\sigma$, if $v$
is not matched by $\ranking(\sigma)$, then $u$ is matched to a vertex
whose rank is less than the rank of $v$ in $\sigma$.
\end{lemma}
\begin{proof}
If $v$ is not matched by $\ranking(\sigma)$, then since there is an
edge between $u$ and $v$, it was available to be matched to $u$ when
it arrived.  Therefore, by the behavior of $\ranking$, $u$ must have
been matched to a vertex of lower rank.
\end{proof}

\begin{lemma}\label{lem:technical}
Let $u \in U_H$, and let $v = f(u)$.  Fix an integer $t$ such that
$1 \leq t \leq n$.  Let $\sigma$ be a permutation, and let $\sigma'$
be the permutation obtained from $\sigma$ by removing vertex $v$ and
putting it back in so its rank is $t$.  If $v$ is not matched by
$\ranking(\sigma')$, then $u$ must be matched by $\ranking(\sigma)$ to
a vertex whose rank in $\sigma$ is less than or equal to $t$.
\end{lemma}
\begin{proof}
For the proof, it is convenient to invoke Lemma~\ref{lem:duality} and
consider $\ranking'(\sigma)$ and $\ranking'(\sigma')$ instead of
$\ranking(\sigma)$ and $\ranking(\sigma')$.  In the process by which
$\ranking'$ constructs its matching, call the moment that the
$t^\textrm{th}$ vertex in $V$ arrives \emph{time} $t$.  For any $1
\leq s \leq n$, let $R_s(\sigma)$ (resp., $R_s(\sigma')$) be the set
of vertices in $U_H$ matched by time $s$ in $\sigma$ (resp.,
$\sigma'$).  By Lemma~\ref{lem:easy}, if $v$ is not matched by
$\ranking(\sigma')$, then $u$ must be matched to a vertex $v'$ in
$\ranking(\sigma')$ such that $\sigma'(v') < \sigma'(v)$.  Hence $u
\in R_{t-1}(\sigma')$.  We prove the lemma by showing that
$R_{t-1}(\sigma') \subseteq R_t(\sigma)$.
Let $\tilde{t}$ be the time that $v$ arrives in $\sigma$.  Then
if $\tilde{t} \geq t$, the two orders $\sigma$ and $\sigma$ are
identical through time $t$, which implies that
$R_{t-1}(\sigma') = R_{t-1}(\sigma) \subseteq R_t(\sigma)$.

Now, in the case that $\tilde{t} < t$, we prove that for $1 \leq s
\leq t$, $R_{s-1}(\sigma') \subseteq R_s(\sigma)$.  The proof, which
is similar to the proof of Lemma~\ref{lem:monotonicity}, proceeds by
induction on $s$.  When $s = 0$, the claim clearly holds, since
$R_0(\sigma') = \emptyset$.  Now, supposing that $R_{s-2}(\sigma')
\subseteq R_{s-1}(\sigma)$, we prove that $R_{s-1}(\sigma') \subseteq
R_s(\sigma)$.  If $s \leq \tilde{t}$, then the two orders $\sigma$ and
$\sigma'$ are identical through time $s$, so $R_{s-1}(\sigma') =
R_{s-1}(\sigma) \subseteq R_s(\sigma)$.  Now suppose that $s >
\tilde{t}$.  Then the vertex that arrives at time $s-1$ in $\sigma'$
is the same as the vertex that arrives at time $s$ in $\sigma$.  Call
this vertex $w$.  If $w$ is not matched by $\ranking'(\sigma')$, then
$R_{s-1}(\sigma') = R_{s-2}(\sigma')$, and we are done by the
induction hypothesis.  Now suppose that $w$ is matched to vertex $x'$
by $\ranking'(\sigma')$ and to vertex $x$ by $\ranking'(\sigma)$.  If
$x = x'$, then again we are done by the induction hypothesis, so
suppose that $x \ne x'$.  Since $x$ was available at time $s-1$ in
$\sigma$, we have $x \not\in R_{s-1}(\sigma)$, and by the induction
hypothesis $x \not\in R_{s-2}(\sigma')$.  Hence, $x$ was available at
time $s-1$ in $\sigma'$.  Since $\ranking'(\sigma')$ matched $w$ to
$x'$, it must be that $\pi(x') < \pi(x)$.  This implies that $x'$ must
be matched when $w$ arrives at time $s$ in $\sigma$, or in other
words, $x' \in R_{s-1}(\sigma) \subseteq R_s(\sigma)$.  By the
induction hypothesis, we are done.
\end{proof}

\begin{lemma}\label{lem:rankingmain}
For $1 \leq t \leq n$, let $x_t$ denote the probability over $\sigma$
that the vertex ranked $t$ in $V$ is matched by $\ranking(\sigma)$.
Then
\begin{equation}\label{eqn:rankingmain}
1 - x_t \leq \frac{1}{kn} \sum_{s = 1}^t x_s \enspace .
\end{equation}
\end{lemma}
\begin{proof}
Let $\sigma$ be permutation chosen uniformly at random, and let
$\sigma'$ be a permutation obtained from $\sigma$ by choosing a vertex
$v \in V$ uniformly at random, taking it out of $\sigma$, and putting
it back so that its rank is $t$.  Note that both $\sigma$ and
$\sigma'$ are distributed uniformly at random among all permutations.
Let $u$ be a vertex chosen uniformly at random from $f^{-1}(v)$.  Note
that conditioned on $\sigma$, $u$ is equally likely to be any of the
$kn$ vertices in $U_H$.  Let $R_t$ be the set of vertices in $U_H$
that are matched by $\ranking(\sigma)$ to a vertex of rank $t$ or
lower in $\sigma$.  Lemma~\ref{lem:technical} states that if $v$ is
not matched by $\ranking(\sigma')$, then $u \in R_t$.  The expected
size of $R_t$ is $\sum_{1 \leq s \leq t} x_s$.  Hence, the probability
that $u \in R_t$, conditioned on $\sigma$, is $(1/(kn)) \sum_{1 \leq s
  \leq t} x_s$.  The lemma follows because the probability that $v$ is
not matched by $\ranking(\sigma')$ is $1 - x_t$.
\end{proof}

We are now ready to prove Theorem~\ref{thm:ranking}.
\begin{proof}[Proof of Theorem~\ref{thm:ranking}]
For $0 \leq t \leq n$, let $S_t = \sum_{1 \leq s \leq t} x_s$.  Then
the expected size of the matching returned by Ranking on $H$ is $S_n$.
Rearranging (\ref{eqn:rankingmain}) yields, for $1 \leq t \leq n$,
\begin{equation*}
S_t \geq \paren{\frac{kn}{kn+1}}\paren{1 + S_{t-1}},
\end{equation*}
which by induction implies that $S_t \geq \sum_{1 \leq s \leq t}
(kn/(kn+1))^s$, and hence
\begin{equation*}
S_n \geq \sum_{s = 1}^n \paren{\frac{kn}{kn+1}}^s
= kn \paren{1 - \paren{1 - \frac{1}{kn+1}}^n}
= kn \paren{1 - \frac{1}{e^{1/k}} + o(1)} \enspace .
\end{equation*}
\end{proof}